\documentclass[pra,aps,nopacs,onecolumn,twoside,superscriptaddress]{revtex4}

\usepackage{amsmath,amsfonts,amssymb,caption,color,epsfig,graphics,graphicx,hyperref,latexsym,mathrsfs,revsymb,theorem,url,verbatim,epstopdf,mathtools,enumerate,float}
\usepackage[T1]{fontenc}
\usepackage{tikz}
\usepackage{graphicx}
\usepackage{varwidth}
\usepackage{algorithmic}
\usepackage{subcaption}

\usetikzlibrary{positioning, shapes.geometric}

\hypersetup{colorlinks,linkcolor={blue},citecolor={blue},urlcolor={red}}

\newtheorem{definition}{Definition}
\newtheorem{proposition}[definition]{Proposition}
\newtheorem{lemma}[definition]{Lemma}

\newtheorem{theorem}[definition]{Theorem}
\newtheorem{corollary}[definition]{Corollary}
\newtheorem{conjecture}[definition]{Conjecture}

\newtheorem{remark}[definition]{Remark}
\newtheorem{example}[definition]{Example}
\newtheorem{question}[definition]{Question}
\newtheorem{memo}[definition]{Memo}

\def\squareforqed{\hbox{\rlap{$\sqcap$}$\sqcup$}}
\def\qed{\ifmmode\squareforqed\else{\unskip\nobreak\hfil
\penalty50\hskip1em\null\nobreak\hfil\squareforqed
\parfillskip=0pt\finalhyphendemerits=0\endgraf}\fi}
\def\endenv{\ifmmode\;\else{\unskip\nobreak\hfil
\penalty50\hskip1em\null\nobreak\hfil\;
\parfillskip=0pt\finalhyphendemerits=0\endgraf}\fi}
\newenvironment{proof}{\noindent \textbf{{Proof.~} }}{\qed}
\def\Dbar{\leavevmode\lower.6ex\hbox to 0pt
{\hskip-.23ex\accent"16\hss}D}
\makeatletter
\def\url@leostyle{%
  \@ifundefined{selectfont}{\def\UrlFont{\sf}}{\def\UrlFont{\small\ttfamily}}}
\makeatother
\def\bcj{\begin{conjecture}}
\def\ecj{\end{conjecture}}
\def\bcr{\begin{corollary}}
\def\ecr{\end{corollary}}
\def\bd{\begin{definition}}
\def\ed{\end{definition}}
\def\bea{\begin{eqnarray}}
\def\eea{\end{eqnarray}}
\def\bem{\begin{enumerate}}
\def\eem{\end{enumerate}}
\def\bex{\begin{example}}
\def\eex{\end{example}}
\def\bim{\begin{itemize}}
\def\eim{\end{itemize}}
\def\bl{\begin{lemma}}
\def\el{\end{lemma}}
\def\bma{\begin{bmatrix}}
\def\ema{\end{bmatrix}}
\def\bpf{\begin{proof}}
\def\epf{\end{proof}}
\def\bpp{\begin{proposition}}
\def\epp{\end{proposition}}
\def\bqu{\begin{question}}
\def\equ{\end{question}}
\def\br{\begin{remark}}
\def\er{\end{remark}}
\def\bt{\begin{theorem}}
\def\et{\end{theorem}}
\def\bmm{\begin{memo}}
\def\emm{\end{memo}}

\def\btb{\begin{tabular}}
\def\etb{\end{tabular}}

\newcommand{\nc}{\newcommand}

 \nc{\bbA}{\mathbb{A}} \nc{\bbB}{\mathbb{B}} \nc{\bbC}{\mathbb{C}}
 \nc{\bbD}{\mathbb{D}} \nc{\bbE}{\mathbb{E}} \nc{\bbF}{\mathbb{F}}
 \nc{\bbG}{\mathbb{G}} \nc{\bbH}{\mathbb{H}} \nc{\bbI}{\mathbb{I}}
 \nc{\bbJ}{\mathbb{J}} \nc{\bbK}{\mathbb{K}} \nc{\bbL}{\mathbb{L}}
 \nc{\bbM}{\mathbb{M}} \nc{\bbN}{\mathbb{N}} \nc{\bbO}{\mathbb{O}}
 \nc{\bbP}{\mathbb{P}} \nc{\bbQ}{\mathbb{Q}} \nc{\bbR}{\mathbb{R}}
 \nc{\bbS}{\mathbb{S}} \nc{\bbT}{\mathbb{T}} \nc{\bbU}{\mathbb{U}}
 \nc{\bbV}{\mathbb{V}} \nc{\bbW}{\mathbb{W}} \nc{\bbX}{\mathbb{X}}
 \nc{\bbZ}{\mathbb{Z}}

 \nc{\bA}{{\bf A}} \nc{\bB}{{\bf B}} \nc{\bC}{{\bf C}}
 \nc{\bD}{{\bf D}} \nc{\bE}{{\bf E}} \nc{\bF}{{\bf F}}
 \nc{\bG}{{\bf G}} \nc{\bH}{{\bf H}} \nc{\bI}{{\bf I}}
 \nc{\bJ}{{\bf J}} \nc{\bK}{{\bf K}} \nc{\bL}{{\bf L}}
 \nc{\bM}{{\bf M}} \nc{\bN}{{\bf N}} \nc{\bO}{{\bf O}}
 \nc{\bP}{{\bf P}} \nc{\bQ}{{\bf Q}} \nc{\bR}{{\bf R}}
 \nc{\bS}{{\bf S}} \nc{\bT}{{\bf T}} \nc{\bU}{{\bf U}}
 \nc{\bV}{{\bf V}} \nc{\bW}{{\bf W}} \nc{\bX}{{\bf X}}
 \nc{\bZ}{{\bf Z}}

\nc{\cA}{{\cal A}} \nc{\cB}{{\cal B}} \nc{\cC}{{\cal C}}
\nc{\cD}{{\cal D}} \nc{\cE}{{\cal E}} \nc{\cF}{{\cal F}}
\nc{\cG}{{\cal G}} \nc{\cH}{{\cal H}} \nc{\cI}{{\cal I}}
\nc{\cJ}{{\cal J}} \nc{\cK}{{\cal K}} \nc{\cL}{{\cal L}}
\nc{\cM}{{\cal M}} \nc{\cN}{{\cal N}} \nc{\cO}{{\cal O}}
\nc{\cP}{{\cal P}} \nc{\cQ}{{\cal Q}} \nc{\cR}{{\cal R}}
\nc{\cS}{{\cal S}} \nc{\cT}{{\cal T}} \nc{\cU}{{\cal U}}
\nc{\cV}{{\cal V}} \nc{\cW}{{\cal W}} \nc{\cX}{{\cal X}}
\nc{\cZ}{{\cal Z}}

\nc{\hA}{{\hat{A}}} \nc{\hB}{{\hat{B}}} \nc{\hC}{{\hat{C}}}
\nc{\hD}{{\hat{D}}} \nc{\hE}{{\hat{E}}} \nc{\hF}{{\hat{F}}}
\nc{\hG}{{\hat{G}}} \nc{\hH}{{\hat{H}}} \nc{\hI}{{\hat{I}}}
\nc{\hJ}{{\hat{J}}} \nc{\hK}{{\hat{K}}} \nc{\hL}{{\hat{L}}}
\nc{\hM}{{\hat{M}}} \nc{\hN}{{\hat{N}}} \nc{\hO}{{\hat{O}}}
\nc{\hP}{{\hat{P}}} \nc{\hR}{{\hat{R}}} \nc{\hS}{{\hat{S}}}
\nc{\hT}{{\hat{T}}} \nc{\hU}{{\hat{U}}} \nc{\hV}{{\hat{V}}}
\nc{\hW}{{\hat{W}}} \nc{\hX}{{\hat{X}}} \nc{\hZ}{{\hat{Z}}}

\nc{\hn}{{\hat{n}}}

\newcommand{\ket}[1]{|#1\rangle}

\begin{document}

\Large

\title{The construction of multiqubit unextendible product bases}

\date{\today}


\author{Caohan Cheng}\email[]{chengch@buaa.edu.cn}
\affiliation{LMIB(Beihang University), Ministry of education, and School of Mathematical Sciences, Beihang University, Beijing 100191, China}
\author{Lin Chen}\email[]{linchen@buaa.edu.cn (corresponding author)}
\affiliation{LMIB(Beihang University), Ministry of education, and School of Mathematical Sciences, Beihang University, Beijing 100191, China}
\affiliation{International Research Institute for Multidisciplinary Science, Beihang University, Beijing 100191, China}

\begin{abstract}
The unextendible orthogonal matrices (UPBs) can be used for various problems in quantum information. We provide an algorithm to check if two UPBs are non-equivalent to each other. We give a method to construct UPBs and we apply this method to find all $5$-qubit UPBs of size eight. We apply the algorithm to check if the $5$-qubit UPBs of size eight are non-equivalent to each other. Based on all the $5$-qubit UPBs of size eight, we propose a theorem for constructing a new UPB non-equivalent to a given one. 
\end{abstract}

\maketitle

\section{introduction}
\label{sec:intro}

Unextendible product basis (UPB) is a set of orthonormal product vectors such that no product vector is orthogonal to all of them \cite{bennett1999unextendible,bennett1999quantum,chen2015minimum}. The UPBs were firstly  used to explain the phenomenon of non-locality without entanglement. Vectors of a UPB cannot be discriminated by local operations and classical communication 
(LOCC) despite the absence of entanglement \cite{bennett1999unextendible}. 
UPBs are powerful tools in quantum information, which have subsequently been used for the construction of positive-partial-transpose (PPT) entangled states \cite{bennett1999quantum,chen2011distillability}. Besides UPBs can be used for the study of Bell inequalities and completely entangled subspaces \cite{demianowicz2018unextendible,bhat2006completely}.  The relevant Bell inequalities that arise from the construction all come from UPBs \cite{PhysRevA.85.042113}. This helps in finding a maximally entangled subspace of maximal dimension \cite{bhat2006completely}. $k-CES$ is completely entangled subspaces of entanglement depth k, one can use UPBs to construct it \cite{wang20194,shi2023graph,PhysRevA.85.042113,sorensen2001entanglement}. A UPB which is unextendible with (k-1)-producible vectors is called a (k-1)-UPB. The complement of a (k-1)-unextendible product basis is a $k-CES$.
The multiqubit states such as Greenberger–Horn–Zeilinger (GHZ) and W states \cite{haffner2005scalable} are easier to realize in experiment. One can use GenTiles2 UPB to construct indecomposable positive maps \cite{terhal2001family,divincenzo2003unextendible}.
However we do not know much about the mathematical structure of the mutiqubit UPBs. 
The 3-qubit UPB in  $ (C^2)^{\otimes 3}$ is named the shift UPB \cite{bravyi2004unextendible}.
One can use a result on 1-factorization of complete graphs to find some cases of the minimum sizes of UPBs \cite{feng2006unextendible}.
The formal matrix approach is first used to study the orthogonal product bases 
(OPB) in multiqubit systems \cite{chen2018nonexistence,chen2017orthogonal}. The formal approach can be used to the study of mutiqubit UPBs, it helps to prove that there do not exist a $n$-qubit UPB of size $2^n - 5$ \cite{chen2017orthogonal}. The formal approach can provide a lot of theorems and ideas to construct a new UPB \cite{chen2018multiqubit}.
\par
  
In section \ref{sec:pre}, we describe the concept of UPBs and UOMs. Then we offer an overview of the current theorems and conclusions of the existence of UPBs. In section \ref{sec:85 upb}, we introduce an algorithm to check the equivalence of two UOMs of the same size. We compare two UOMs to determine if they are equivalent by performing row/column permutations and renaming. We realize it by putting the rows and columns of these two UOMs in the same order and compare them. We offer a table FIG \ref{fig:flowchart2} to describe the specific details of the algorithm. In section \ref{sec:algotithm}, we give all the 5-qubit UOMs of size eight. Then we introduce the method we apply to find them in FIG \ref{fig:flowchart}. Linear integer Diophantine problems now have some efficient algorithms \cite{frumkin1977polynomial,aardal2000solving,ruzsa1993solving}. We transform the process of finding small-scale UPBs into the form of solving small-scale linear Diophantine equations. Then we can construct UOMs by solving the equations. 
\par
The rest of this paper is organized as follows. In section \ref{sec:pre}, we provide the lemmas and theorems we apply. We offer some examples to describe them. In section \ref{sec:algotithm}, we provide Theorem \ref{the:2123} based on all the $5$-qubit UPBs. This theorem helps in obtaining a new UOM non-equivalent to a given one. Finally, we conclude in \ref{sec:con}. 

\section{preliminaries}

\label{sec:pre}

In this section, we introduce the preliminary results. In subsection \ref{sub:defin of upb and uom}, we describe the concept of UPBs and UOMs. Then we make a overview about the current main conclusions on constructing UPBs. In subsection \ref{sub:lemma}, we introduce the lemmas and theorems we use in the rest of this paper.
\subsection{Known mutiqubit UPBs}
\label{sub:defin of upb and uom}
In quantum information, an n-partite unextendible product basis (UPB) in $C^{d_1} \otimes  ... \otimes C^{d_n}$ is a set of n-partite orthonormal product vectors such that no  n-partite product vector is orthogonal to all of them \cite{wang2020construction}. We require that the cardinality of the set is smaller than ${d_1} {d_2} ... {d_n}$.
For convenience, we define the concept of unextendible orthogonal matrix (UOM). We express the product vectors of  an $n$-qubit UPB with size $m$ as the row vectors in an $m\times n$ matrix, respectively. We refer to the matrix as a UOM, and any two row vectors of a UOM are orthogonal to each other. For the orthogonal qubit vectors $\ket{a_{11}}$ and $\ket{a_{11}'}$, we define them as the vector variables ${a_{11}}$ and ${a_{11}'}$ in a UOM.
The 4-qubit UPB\\ $\{\ket{a_{11},a_{12},a_{13},a_{14}}, \ket{a_{11},a_{12},a_{13},a_{14}'}, \ket{a_{11},a_{12},a_{13},a_{14}'}, \ket{a_{11},a_{12},a_{13}',a_{34}',}, \ket{a_{11}',a_{12},a_{53},a_{14}},\\\ket{a_{61},a_{12}',a_{13}',a_{64}}, \ket{a_{11},a_{12}',a_{13}',a_{64}}, \ket{a_{61}',a_{12}',a_{53},a_{14}},
\ket{a_{11}',a_{92},a_{13},a_{14}'}, \ket{a_{11}',a_{92}',a_{53}',a_{64}'}\}$ can be expressed as
\begin{equation}\label{ori_matrix}
	\begin{pmatrix}
		a_{11} & a_{12} & a_{13} & a_{14}\\
		a_{11} & a_{12} & a_{13} & a_{14}'\\
		a_{11} & a_{12} & a_{13}' & a_{34}\\
		a_{11} & a_{12} & a_{13}' & a_{34}'\\
		a_{11}' & a_{12} & a_{53} & a_{14}\\
		a_{61} & a_{12}' & a_{13}' & a_{64}\\
		a_{11} & a_{12}' & a_{53}' & a_{64}'\\
		a_{61}' & a_{12}' & a_{53} & a_{14}\\
		a_{11}' & a_{92} & a_{13} & a_{14}'\\
		a_{11}'& a_{92}' & a_{53}' & a_{64}'
	\end{pmatrix}.
\end{equation}\par
In the following, we review the known multiqubit UPBs so far \cite{tura2012four}.  Researchers have exhausted all n-qubits UPBs up to $n = 8$, see  Table \ref{ta:UOMS findings_1} - \ref{ta:UOMS findings_3}. As there is no two-qubit UPB, the three-qubit UPB was the first multiqubit UPB studied by the Shifts UPBs \cite{bennett1999unextendible}. All four-qubit UPBs, the five-qubit UPBs (except the sizes $11$,$27$), the six-qubit UPBs (except the sizes $10$,$11$,$13$,$59$), the seven-qubit UPBs (except the sizes 10, 11, 13, 14, 15, 19, 123) were constructed by \cite{johnston2014structure}. Further, all eight-qubit UPBs were constructed by \cite{chen2018multiqubit}. 
In \cite{johnston2014structure}, Johnston proposed a theorem solving the existence of almost all mutiqubit UPBs whose size $s$ is a multiple of $4$, except the case handling n-qubit UPBs and $n \equiv1 \pmod{4}$ and its size $s = 2n$.  Ref \cite{johnston2014structure,chen2018nonexistence} proposed that for any $n \geq 4$, the n-qubits UPBs do not exist of sizes $2^n - 1$, $2^n - 2$, $2^n - 3$ and $2^n - 5$. Then the  remaining n-qubit UPBs of $n \in \left[4,7\right]$ were discovered and constructed by \cite{chen2018multiqubit,wang2020construction}.
The minimum size of n-qubit UPBs, denoted as $f\left(n \right)$, is shown in \cite{johnston2013minimum,wang2020construction} with
\begin{equation}
	f(n) = \begin{cases}
		& n+1,\text{ if  n is odd; } \\
		& n+2,\text{ if  n = 4 or } n\equiv 2\left ( mod\ 4 \right ) \\
		& n+3,\text{ if  n is 8; } \\
		& n+4,\text{ otherwise. }
	\end{cases}
\end{equation}
\par
Ref \cite{johnston2014structure} proposed that for any odd $n$, the n-qubit UPB of size $n + 2$ does not exist. For example, the 9-qubit UPB of size $11$ and the 11-qubit of size $13$ do not exist.
\begin{table}[H]
	\begin{center}
		\caption{Statistics of UOMs for $n = 3 - 8$. The letter $e$ means the UOMs exist in $\mathcal{O}$(m,n). The number 0 means the UOMs do not exist in $\mathcal{O}$(m,n). }
		\label{ta:UOMS findings_1}
		\begin{tabular}{l|cccccccccccc}
			\hline
			$m$/$n$ & $3$ &$4$ &$5$ &$6$ &$7$ &$8$\\
			\hline
			3 & 0 & 0 & 0 & 0 & 0 & 0\\
			4 & e & 0 & 0 & 0 & 0 & 0\\
			5 & 0 & 0 & 0 & 0 & 0 & 0\\
			6 & 0 & e & e & 0 & 0 & 0\\
			7 & 0 & e & 0 & 0 & 0 & 0\\
			8 & e & e & e & e & e & 0\\
			9 & 0 & e & 0 & 0 & 0 & 0\\
			10 & 0 & e & e & e & e & 0\\
			11 & 0 & 0 & e & e & e & e\\
			12 & 0 & e & e & e & e & e\\
			13-15 & 0 & 0 & e & e & e & e\\
			16 & 0 & e & e & e & e & e\\
			17-26 & 0 & 0 & e & e & e & e\\
			27 & 0 & 0 & 0 & e & e & e\\
			28 & 0 & 0 & e & e & e & e\\
			29-31 & 0 & 0 & 0 & e & e & e\\
			32 & 0 & 0 & e & e & e & e\\
			33-58 & 0 & 0 & 0 & e & e & e\\
			59 & 0 & 0 & 0 & 0 & e & e\\
			60 & 0 & 0 & 0 & e & e & e\\
			61-63 & 0 & 0 & 0 & 0 & e & e\\
			64 & 0 & 0 & 0 & e & e & e\\
			65-122 & 0 & 0 & 0 & 0 & e & e\\
			123 & 0 & 0 & 0 & 0 & 0 & e\\
			124 & 0 & 0 & 0 & 0 & e & e\\
			125-127 & 0 & 0 & 0 & 0 & 0 & e\\
			128 & 0 & 0 & 0 & 0 & e & e\\
			129-250 & 0 & 0 & 0 & 0 & 0 & e\\
			251 & 0 & 0 & 0 & 0 & 0 & 0\\
			252 & 0 & 0 & 0 & 0 & 0 & e\\
			253-255 & 0 & 0 & 0 & 0 & 0 & 0\\
			256 & 0 & 0 & 0 & 0 & 0 & e\\
			
			\hline			
		\end{tabular}
	\end{center}
\end{table}
\begin{table}[h!]
	\begin{center}
		\caption{Statistics of UOMs for $n = 9 - 13$. Known UOMs in $\mathcal{O}$(m,n). The number with "'" means that it has been confirmed that it does not exist.}
		\label{ta:UOMS findings_2}
		\begin{tabular}{l|lllll}
			\hline
			$n$ & \hspace{16em}$m$\\
			\hline
			9 & \hspace{5em}     10, 11', 12, 16, 22-506, 507', 508, 509', 510', 511'\hspace{1em}512\\
			10 & \hspace{5em}       12, 16, 20, 22, 24, 26, 28, 32-1018, 1019', 1020, 1021', 1022', 1023',\hspace{1em}1024\\
			11 & \hspace{5em}       12, 13', 16, 20, 24, 28, 32, 34, 36, 38, 40, 42, 44-2042, 2043', 2044, 2045', 2046', 2047', 2048\\
			12 & \hspace{5em}       16, 20, 24, 28, 32, 36, 40, 44, 46, 48, 50, 52, 54, 56-4090, 4091', 4092, 4093', 4094', 4095', 4096\\
			13 & \hspace{5em}       14, 15', 16, 20, 24, 32, 36, 40, 44, 48, 52, 56, 60, 62, 64, 66, 68, 70, 72-8186, 8187', 8188, 8189', 8190', 8191', 8192\\
			\hline			
		\end{tabular}
	\end{center}
\end{table}
\begin{table}[h!]
	\begin{center}
		\caption{Statistics of UOMs current findings(iii). Unknown UOMs in $\mathcal{O}$(m,n).}
		\label{ta:UOMS findings_3}
		\begin{tabular}{l|lllll}
			\hline
			$n$ & \hspace{16em}$m$\\
			\hline
			9 & \hspace{5em}     13-15,\hspace{2em}17-21\\
			10 & \hspace{5em}       13-15, 17-19, 21, 23, 25, 27, 29-31\\
			11 & \hspace{5em}       14-15, 17-19, 21-23, 25-27, 29-31, 33, 35, 37, 39, 41, 43\\
			12 & \hspace{5em}       17-19, 21-23,25-27,29-31,32-35,37-39,41-43,45,47,49,51,53,55\\
			13 & \hspace{5em}       17-19,21-23,25-31,33-35,37-39,41-43,45-47,49-51,53-55,57-59,61,63,65,67,69,71\\
			\hline			
		\end{tabular}
	\end{center}
\end{table}
\subsection{Notations and properties of UOMs}
\label{sub:lemma}
For ease of programming, we define the orthogonality of UPB elements.
We define that two elements in a column are orthogonal if and only if they are an odd number $i$ and an even number $j$, and $j = i + 1$ \cite{wang2020construction}.
We rename UPBs according to the UPBs element orthogonality. We firstly rename the top element as $1$, the same element in this column as $1$, and the element orthogonal to this element as $2$. After that, we name the top unnamed element as $3$, the same element in this column as $3$, and the element orthogonal to this element as $4$, and so on. Each time we name the top unnamed element with the smallest unused positive integer in this column.\par
	The UOM \eqref{ori_matrix} is expressed as a UOM $A$, the notation $a_{11}'$ means it is orthogonal to $a_{11}$. The UPB in \eqref{ori_matrix} can be renamed according to the orthogonality property. Taking the first column for example, since $a_{11}$ appears at the top, we name it as $1$ and the orthogonal $a_{11}'$ is named as $2$. Since $a_{61}$ is the topmost except $a_{11}$ and $a_{11}'$, we name it as $3$, and the orthogonal $a_{61}'$ is named as  $4$. Finally we can express \eqref{ori_matrix} as follows.\\
\begin{equation}\label{matrix A}
	A:=\begin{pmatrix}
	1& 1& 1& 1\\
	1& 1& 1& 2\\
	1& 1& 2& 5\\
	1& 1& 2& 6\\
	2& 1& 3& 1\\
	3& 2& 2& 3\\
	1& 2& 4& 4\\
	4& 2& 3& 1\\
	2& 3& 1& 2\\
	2& 4& 4& 4
	\end{pmatrix}.\\	
\end{equation}
\par
Next we present a property of UOMs we apply in the algorithm in section \ref{sec:85 upb}. 
\begin{theorem}\label{le:21}
	The elements in the same column of a UOM are still in the same column after row/column permutation and naming. The elements in the same row of a UOM are still in the same row after row/column permutation and naming.\par
\end{theorem}\par
The theorem can be proofed by the definition of UOMs. 
	We define certain quantities and methods to handle matrices we obtain after renaming uniformly.
    We define row-wise sorting. It is a way to sort matrix rows.
	First, we place the numbers in the first column of A as follows. We place them in the ascending  order from top to bottom. Now the first column of A consists of $d_1$ ones, $d_2$ twos, ..., $d_{2k-1}$ $(2k-1)$s, and $d_{2k}$ $(2k)s$.
	To permute the second column of A, we consider  the left-upper $d_1\times2$ submatrix of A. We place the numbers in the second column of the submatrix in the ascending order from top to bottom.  Now the $d_1\times2$ submatrix consists of submatices whose row numbers are $e_1$, $e_2$, ..., $e_i$, such that $e_1$ + $e_2$+ ... + $e_i$ = $d_1$. One can similarly treat other submatrices in A, such as $d_2\times2$ submatrix, ..., $d_{2k}\times2$ submatrix. To permute the third column of A, we handle the second and the third column of A in the way we handled  the first and second columns of A. 
	After permuting all columns of A , we obtain A after row-wise.
	sorting.
	We define the multiplicity of an element in a column. It is the number of times it appears in that column.
	We sort rows of $A$  in \eqref{matrix A}. We obtain a new matrix $A-sorted$
\begin{equation}\label{A-sort}
	\begin{pmatrix}
		1& 1& 1& 1\\
		1& 1& 1& 2\\
		1& 1& 2& 5\\
		1& 1& 2& 6\\
		1& 2& 4& 4\\
		2& 1& 3& 1\\
		2& 3& 1& 2\\
		2& 4& 4& 4\\
		3& 2& 2& 3\\
		4& 2& 3& 1
	\end{pmatrix}.
\end{equation}\\
	
	We calculate the  multiplicity of elements in the first column of $A$ in \eqref{matrix A}. One can see that the multiplicity of $1,2,3,4$  in the first column of $A$ in \eqref{matrix A} are $5,3,1,1$, respectively.\par
	 We use the notation $\omega \left ( A_j \right )$ to express the column feature of the jth column of A. $A_j$ is the j-th column of the matrix $A$. The matrix $\omega \left ( A_j \right )$ has two columns, and the row number is the count of orthogonal element pairs in the column.The two elements in each row represents the multiplicities of each orthogonal element of a orthogonal pair in the column. To fix $\omega \left ( A_j \right )$, we give a restriction on $\omega \left ( A_j \right )$.
	  Each row of $\omega \left ( A_j \right )$ is in the ascending order from left to right. The matrix $\omega \left ( A_j \right )$ does not change if we sort it. We store the standard column features of all columns of the matrix in a list and refer to the list as "column feature list".
	 We calculate the column feature and the standard column feature of the first column of $A$ in \eqref{matrix A}. We will calculate the column feature list of $A$ in \eqref{matrix A}.
	When calculating, we agree to start from the smaller number. The first column of A has five ones, three twos are orthogonal to ones, one 3 that is not orthogonal to 1 or 2, and one 4 that is orthogonal to 3.
	We can obtain the column feature of the first column of $A$ in \eqref{matrix A} is $
	\begin{pmatrix}
		5& 3\\
		1& 1\\
	\end{pmatrix}$.
	We can respectively obtain the standard column feature of the first column of $A$ in \eqref{matrix A} is 
$
	\begin{pmatrix}
		1& 1\\
		3& 5\\
	\end{pmatrix}$ and the column feature list of $A$ in \eqref{matrix A} is 

\begin{equation}
	\left[\left(\begin{array}{ll}
		1 & 1 \\
		3 & 5
	\end{array}\right),\left(\begin{array}{ll}
		1 & 1 \\
		3 & 5
	\end{array}\right),\left(\begin{array}{ll}
		2 & 2 \\
		3 & 3
	\end{array}\right),\left(\begin{array}{ll}
		1 & 1 \\
		1 & 2 \\
		2 & 3
	\end{array}\right)\right].\\
\end{equation}

We define the multiplicity of an element in a column. It is the number of times it appears in that column.
One can see that the multiplicity of $1,2,3,4$  in the first column of $A$ are $5,3,1,1$, respectively.
So we can define a set $\left \{ a_{1}^{\left ( i \right ) }, b_{1}^{\left ( i \right ) },...,a_{n_{i}}^{\left ( i \right ) },b_{n_{i}}^{\left ( i \right ) }  \right \} _{A}$ correspond to a column of a given matrix $A$ \cite{johnston2014structure}. The number $i$ means it is the ith column of $A$. The number $n_{i}$ is the total number of orthogonal pairs in the i-th column of $A$. So we obtain two theorem used in section \ref{sec:algotithm}
\begin{theorem}\label{eq:5}
\begin{equation}\label{eq:4}
	{\textstyle \sum_{i=1}^{n_{i}}}(a_{n_{i}}^{\left (i \right ) } + b_{n_{i}}^{\left ( i \right ) }) = m,
\end{equation}
	\begin{equation}
		{\textstyle \sum_{i,j}}a_{j}^{\left(i\right)}b_{j}^{\left(i\right)}\ge s(s-1)/2.  
	\end{equation}
\end{theorem}

We choose a random element $w_{1}$ in the set $\left \{ a_{1}^{\left ( 1 \right ) }, b_{1}^{\left ( 1 \right ) },...,a_{n_{i}}^{\left ( 1 \right ) },b_{n_{i}}^{\left ( 1 \right ) }  \right \} _{A}$, then we get a new matrix $A'$ by deleting the rows and columns related to $w_{1}$ in $A$. Then we choose a random element $w_{2}$ in the set $\left \{ a_{1}^{\left ( 1 \right ) }, b_{1}^{\left ( 1 \right ) },...,a_{n_{i}}^{\left ( 1 \right ) },b_{n_{i}}^{\left ( 1 \right ) }  \right \} _{A'}$, then we get a new matrix $A"$ by deleting the rows and columns related to $w_{2}$ in $A'$ ,... Finally, we get a sequence $\left\{ w_{i} \right\}$. A given $m\times n$ orthogonal matrix is an UOM if and only if any sequence $\left\{ w_{i} \right\}$ built from it satisfies\par
\begin{theorem}\label{the:UOM}
	\begin{equation}\label{eq:the}
		max\left \{  {\sum_{1}^{i}w_{i}}\right \} \le m-n+i-1.
	\end{equation}
\end{theorem}
One can apply Theorem \ref{the:UOM} to check if a given matrix $M$ is a UOM. If there is one sequence $\left\{ w_{i} \right\}$ built from $M$ and the sequence do not satisfy the inequality \eqref{eq:the}, then $M$ is not a UOM.          

\section{The algorithm used for verifying UPB equivalence}
\label{sec:85 upb}
	In this section, we design an algorithm to check the equivalence of two UOMs of the same size. We use this method to check all the non-equivalent $8\times5$ UPBs in section \ref{sec:algotithm}. If two UOMs are not equivalent, the program will output "not equivalent". If two UOMs are equivalent, the program will output "equivalent" and provide the same matrix obtained by row and column permutation and renaming for these two matrices.\par
	Based on Theorem \ref{le:21}, we locate rows in a UOM with two given properties of rows. The first property is that each of the rows has the first element of minimum multiplicity in the first column. We store these rows in a list. This elements are still in this row after permutation and renaming. So we can locate the position of these rows before and after permutation and naming.The second property is that each of the rows has the most numbers of elements such that each element is larger than or equal to the given "$i$". We store these rows in a list. We construct an example to demonstrate how we locate rows by these two properties.
	 We can obtain rows in \eqref{matrix A} with the first given property.
	 We can see that in the first column of $A$, the multiplicity of \ $3$ and $4$ is the minimum, which corresponds to the two rows in $A$ in \eqref{matrix A} is
	$
	\left[\begin{pmatrix}
		3 & 2 & 2 & 3
	\end{pmatrix},\begin{pmatrix}
		4 & 2 & 3 & 1
	\end{pmatrix} \right ]
	$.
	After the permutation and renaming of $A$ in \eqref{matrix A}, the minimum element in this column of $A$ in \eqref{matrix A} still corresponds to these two rows.
	We obtain rows in \eqref{matrix A} with the second given property .
	Here, we take the given number $i = 3$.
	We can see that the tenth row of $A$ have the most elements larger than or equal to $3$ in $A$ in \eqref{matrix A} is
	$
	\left [ \begin{pmatrix}
		2 & 4 & 4 & 4
	\end{pmatrix} \right ] 
	$.
	 Next, we will provide the details of the steps of the algorithm in FIG \ref{fig:flowchart2}.
	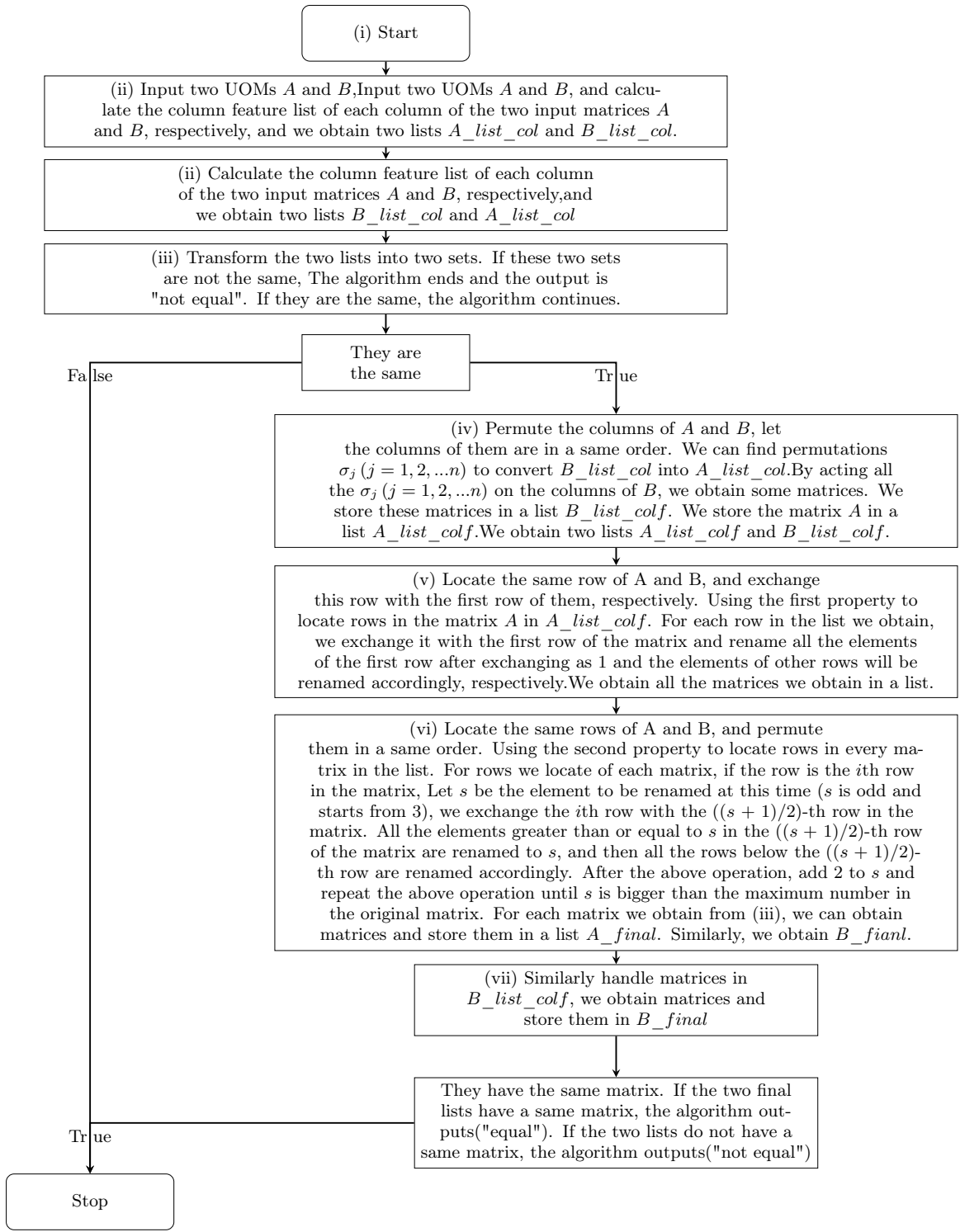
\begin{figure}[H]\label{fg:2}
		\centering
		\begin{minipage}[t]{0.8\textwidth}
			       \scalebox{0.9}{%
			\begin{tikzpicture}[node distance=2cm, auto, on grid]
				\tikzset{
					startstop/.style={
						rectangle,
						rounded corners,
						minimum width=3cm,
						minimum height=1cm,
						align=center,
						draw=black,
						fill=white!30
					},
					io/.style={
						trapezium,
						trapezium left angle=90,
						trapezium right angle=90,
						minimum width=3cm,
						minimum height=1cm,
						align=center,
						draw=black,
						fill=white!30
					},
					process/.style={
						rectangle,
						minimum width=3cm,
						minimum height=1cm,
						align=center,
						draw=black,
						fill=white!30
					},
					decision/.style={
						diamond,
						minimum width=3cm,
						minimum height=1cm,
						align=center,
						draw=black,
						fill=white!30
					},
					arrow/.style={
						thick,
						->,
						>=stealth
					}
				}
				
				\node (start) [startstop] {(i) Start};
				\node (input) [io, below=1.4cm of start,text width = 12cm] {(ii) Input two UOMs $A$ and $B$,Input two UOMs $A$ and $B$, and calculate the column feature list of each column of the two input matrices $A$ and $B$, respectively, and we obtain two lists $A$\_$list$\_$col$ and $B$\_$list$\_$col$.};
				\node (process1) [process, below=1.5cm of input, text width = 12cm] {(ii) Calculate the column feature list of each column \\ of the two input matrices $A$ and $B$, respectively,and \\ we obtain two lists $B$\_$list$\_$col$ and  $A$\_$list$\_$col$};
				\node (process2) [process, below=1.5cm of process1, text width = 12cm] {(iii) Transform the two lists into two sets. If these two sets are not the same, The algorithm ends and the output is "not equal". If they are the same, the algorithm continues.};
				\node (decision) [process, below=1.5cm of process2] {They are \\ the same};
				\node (process2a) [process, below right= 3cm of decision, xshift=2cm, text width = 12cm]  {(iv) Permute the columns of $A$ and $B$, let \\the columns of them are in a same order.  We can find permutations $\sigma_j \left ( j = 1,2,...n  \right )$ to convert $B$\_$list$\_$col$ into $A$\_$list$\_$col$.By acting all the $\sigma_j \left ( j = 1,2,...n \right )$ on the columns of $B$, we obtain some matrices. We store these matrices in a list $B$\_$list$\_$colf$. We store the matrix $A$ in a list $A$\_$list$\_$colf$.We obtain two lists $A$\_$list$\_$colf$ and $B$\_$list$\_$colf$.};
				\node (process2c) [process, below= 2.7cm of process2a, text width = 12cm] {(v) Locate the same row of A and B, and exchange\\ this row with the first row of them, respectively. Using the first property to locate rows in the matrix $A$ in $A$\_$list$\_$colf$. For each row in the list we obtain, we exchange it with the first row of the matrix and rename all the elements of the first row after exchanging as $1$ and the elements of other rows will be renamed accordingly, respectively.We obtain all the matrices we obtain in a list.};
				\node (process2d) [process, below= 3.6cm of process2c, text width = 12cm] {(vi) Locate the same rows of A and B, and permute\\ them in a same order. Using the second property to locate rows in every matrix in the list. For rows we locate of each matrix, if the row is the $i$th row in the matrix, Let $s$ be the element to be renamed at this time ($s$ is odd and starts from $3$), we exchange the $i$th row with the $((s+1)/2)$-th row in the matrix. All the elements greater than or equal to $s$ in the $((s+1)/2)$-th row of the matrix are renamed to $s$, and then all the rows below the $((s+1)/2)$-th row are renamed accordingly. After the above operation, add $2$ to $s$ and repeat the above operation until $s$ is bigger than the maximum number in the original matrix. For each matrix we obtain from (iii), we can obtain  matrices and store them in a list $A\_final$. Similarly, we obtain $B\_fianl$.};
				\node (process2e) [process, below= 3cm of process2d,text width = 7cm] {(vii) Similarly handle matrices in \\$B$\_$list$\_$colf$, we obtain matrices and\\ store them in $B\_final$};
				\node (discussion2) [process, below = 2.2cm of process2e,text width = 7cm] {They have the same matrix. If the two final lists have a same matrix, the algorithm outputs("equal"). If the two lists do not have a same matrix, the algorithm outputs("not equal")};
				\node (stop) [startstop, below left = 2cm of discussion2,xshift = -8cm] {Stop};
				
				\draw [arrow] (start) -- (input);
				\draw [arrow] (input) -- (process1);
				\draw [arrow] (process1) -- (process2);
				\draw [arrow] (process2) -- (decision);
				\draw [arrow] (decision) -| node[anchor=north] {Fa\ lse} (stop);
				\draw [arrow] (decision) -| node[anchor=north] {Tr\ ue} (process2a);
				\draw [arrow] (process2a) -- (process2c);
				\draw [arrow] (process2c) -- (process2d);
				\draw [arrow] (process2d) -- (process2e);
				\draw [arrow] (process2e) -- (discussion2);
				\draw [arrow] (discussion2) -| node[anchor=north] {Tr\ ue} (stop);
			\end{tikzpicture}
		}
	\end{minipage}%
	
		\captionsetup{justification=centering}
		\caption{Program Flowchart and Step Explanation.}
		\label{fig:flowchart2}
	\end{figure}

\begin{theorem}
	Two UOMs are equivalent if and only if there are two identical matrices in the two lists finally obtained by the algorithm.
\end{theorem}
\begin{proof}
	First, we prove that if the two lists $A\_final$ and $B\_final$ we obtain from the algorithm have the same matrix $C$, the two UOMs are equivalent. 
	The two UOMs can obtain $C$ by row/column permutation and renaming, which means that they are equivalent to $C$, so they are equivalent to each other.
	Next, we prove that these two UOMs are equivalent and then there exists a matrix in the final two lists that is the same. 
	In the algorithm, we first replace the columns of the two matrices at the same position through the column feature list. When UOM is equivalent, this can be done.\par Then, we make the rows of the two matrices in the same position by finding the rows with the same elements with special properties in the two matrices. Although the rows obtained are not unique when selecting rows, since we traverse each case each time we select rows, and the two UOMs we operate are equivalent, there is at least one case that the rows of these two matrices are in the same position. In addition, when we select rows, we uniformly rename the selected rows. Therefore, there are at least one matrix in each of the last two lists. These two matrices are obtained by the two input equivalent UOMs, which are transformed by rows and columns permutation so that rows and columns are in the same position, and renamed so that the names are completely consistent. Therefore, the two matrices are exactly the same.
	
\end{proof}
The following two UOMs have been checked by the algorithm we introduce in these section. They are non-equivalent to the known same size UOMs. Suppose $x$ refers to a qubit $\ket{x}$. Next, $x$ and $x'$ are orthonormal qubits $\ket{x}$ and $\ket{x'}$. These two $7$-qubit UPBs are as follows.
\begin{example}
\label{ex:algorithm}
		\label{thm:uom14x7+15x7}
		(i) There is a $14\times 7$ UOM 
		$U_{14,7}=\bma 
		a_{11} & a_{12} & a_{13} & a_{14} & a_{15} & a_{16} & a_{17}\\
		a_{11} & a_{22} & a_{13'} & a_{24} & a_{25} & a_{26} & a_{27}\\
		a_{31} & a_{12} & a_{13'} & a_{24'} & a_{35} & a_{36} & a_{37}\\
		a_{31} & a_{22} & a_{13} & a_{44} & a_{15'} & a_{46} & a_{47}\\
		a_{11'} & a_{52} & a_{53} & a_{54} & a_{35'} & a_{46'} & a_{57}\\
		a_{11'} & a_{52'} & a_{63} & a_{64} & a_{65} & a_{36'} & a_{47'}\\
		a_{31'} & a_{52'} & a_{53} & a_{64'} & a_{75} & a_{26'} & a_{17'}\\
		a_{31'} & a_{52} & a_{63} & a_{14'} & a_{25'} & a_{86} & a_{57'}\\
		a_{91} & a_{22'} & a_{63'} & a_{54'} & a_{75'} & a_{16'} & a_{37'}\\
		a_{91'} & a_{12'} & a_{53'} & a_{44'} & a_{65'} & a_{26'} & a_{57}\\
		a_{91'} & a_{12'} & a_{63'} & a_{64} & a_{25'} & a_{46'} & a_{57'}\\
		a_{11} & a_{12'} & a_{63} & a_{44'} & a_{25} & a_{26} & a_{27'}\\
		a_{11'} & a_{52} & a_{63} & a_{44'} & a_{35'} & a_{86'} & a_{57'}\\
		a_{31'} & a_{52'} & a_{63} & a_{64'} & a_{25'} & a_{26} & a_{17'}\\
		\ema$.
		
		(ii) There is a $15\times 7$ UOM 
		$U_{15,7}=\bma 
		a_{11} & a_{12} & a_{13} & a_{14} & a_{15} & a_{16} & a_{17}\\
		a_{11} & a_{22} & a_{13'} & a_{24} & a_{25} & a_{26} & a_{27}\\
		a_{31} & a_{12} & a_{13'} & a_{24'} & a_{35} & a_{36} & a_{37}\\
		a_{31} & a_{22} & a_{13} & a_{44} & a_{15'} & a_{46} & a_{47}\\
		a_{11'} & a_{52} & a_{53} & a_{54} & a_{35'} & a_{46'} & a_{57}\\
		a_{11'} & a_{52'} & a_{63} & a_{64} & a_{65} & a_{36'} & a_{47'}\\
		a_{31'} & a_{52'} & a_{53} & a_{64'} & a_{75} & a_{26'} & a_{17'}\\
		a_{31'} & a_{52} & a_{63} & a_{14'} & a_{25'} & a_{86} & a_{57'}\\
		a_{91} & a_{22'} & a_{63'} & a_{54'} & a_{75'} & a_{16'} & a_{37'}\\
		a_{91'} & a_{12'} & a_{53'} & a_{44'} & a_{65'} & a_{26'} & a_{57}\\
		a_{91'} & a_{12'} & a_{63'} & a_{64} & a_{25'} & a_{46'} & a_{57'}\\
		a_{11} & a_{12'} & a_{63} & a_{44'} & a_{25} & a_{26} & a_{27'}\\
		a_{31'} & a_{52'} & a_{63} & a_{64'} & a_{25'} & a_{26} & a_{17'}\\
		a_{11'} & a_{52} & a_{63} & a_{24} & a_{15} & a_{86'} & a_{57'}\\
		a_{31'} & a_{52} & a_{63} & a_{14} & a_{15'} & a_{26'} & a_{57'}\\
		\ema$.
	
\end{example}
These two UOMs can be proofed that they are non-equivalent to the knowm same size UOMs . This shows the correctness of the algorithm.

\section{all $8\times 5$ uoms}
\label{sec:algotithm}
In this section, we first present all the $5$-qubit UOMs of size eight in \eqref{eq:41} - \eqref{54}. We have used the algorithm in section \ref{sec:85 upb} to check they are non-equivalent to each other. There are a total of 32 UOMs. Then we introduce the method to construct them in FIG \ref{fig:flowchart}. 
We give all the UOMs and classify them by their combinations of multiplicity. The list $\left[a_1:a_2,...,a_{2j+1} :a_{2j+2}\right]$ presents the combinations of multiplicity in one column. The matrices below the lists are the UOMs relate to them.\\\\
$\left [ 2:2,2:2 \right ] , \left [ 2:2,2:2 \right ] ,\left [ 1:1,1:1,1:1,1:1 \right ] , \left [ 1:1,1:1,1:1,1:1 \right ] ,$\\$ \left [ 1:1,1:1,1:1,1:1 \right ] $:\\
\begin{equation}\label{eq:41}
	\begin{pmatrix}
		1&  1&  1&  1& 1\\
		1&  1&  2&  3& 3\\
		2&  3&  3&  5& 5\\
		2&  3&  4&  7& 7\\
		3&  4&  5&  4& 2\\
		3&  4&  6&  2& 4\\
		4&  2&  7&  8& 6\\
		4&  2&  8&  6& 8
	\end{pmatrix},
	\begin{pmatrix}
		1&  1&  1&  1& 1\\
		1&  1&  3&  2& 3\\
		2&  3&  5&  3& 5\\
		2&  3&  6&  5& 7\\
		3&  4&  2&  7& 4\\
		3&  4&  4&  8& 2\\
		4&  2&  7&  4& 8\\
		4&  2&  8&  6& 6
	\end{pmatrix},
	\begin{pmatrix}
		1&  1&  1&  1& 1\\
		1&  1&  3&  3& 3\\
		2&  3&  5&  5& 5\\
		2&  3&  7&  7& 6\\
		3&  4&  2&  4& 7\\
		3&  4&  4&  2& 8\\
		4&  2&  8&  6& 2\\
		4&  2&  6&  8& 4
	\end{pmatrix},
\end{equation}
\begin{equation}
	\begin{pmatrix}
		1&  1&  1&  1& 1\\
		1&  2&  3&  3& 3\\
		2&  3&  5&  5& 5\\
		2&  4&  7&  7& 7\\
		3&  3&  4&  2& 6\\
		3&  4&  2&  4& 8\\
		4&  1&  8&  6& 2\\
		4&  2&  6&  8& 4
	\end{pmatrix},
	\begin{pmatrix}
		1&  1&  1&  1& 1\\
		1&  2&  3&  3& 3\\
		2&  3&  5&  5& 5\\
		2&  4&  7&  7& 7\\
		3&  3&  2&  6& 4\\
		3&  4&  4&  8& 2\\
		4&  1&  8&  2& 6\\
		4&  2&  6&  4& 8
	\end{pmatrix},
	\begin{pmatrix}
		1&  1&  1&  1& 1\\
		1&  2&  3&  3& 3\\
		2&  3&  5&  5& 5\\
		2&  4&  7&  7& 7\\
		3&  3&  4&  2& 6\\
		3&  4&  4&  6& 2\\
		4&  2&  6&  7& 8\\
		4&  2&  8&  8& 6
	\end{pmatrix}.
\end{equation}\\
$\left [ 2:2,2:2 \right ] , \left [ 2:2,2:2 \right ] ,\left [ 2:2,2:2 \right ] , \left [ 1:1,1:1,1:1,1:1 \right ] , \left [ 1:1,1:1,1:1,1:1 \right ] $:
\begin{equation}\label{42}
	\begin{pmatrix}
		1&  1&  1&  1& 1\\
		1&  2&  3&  3& 3\\
		2&  3&  4&  5& 5\\
		2&  4&  2&  7& 7\\
		3&  4&  4&  2& 8\\
		3&  3&  2&  4& 6\\
		4&  2&  1&  6& 4\\
		4&  1&  3&  8& 2
	\end{pmatrix},
	\begin{pmatrix}
		1&  1&  1&  1& 1\\
		1&  2&  2&  3& 3\\
		2&  3&  3&  5& 5\\
		2&  4&  4&  7& 7\\
		3&  4&  3&  2& 4\\
		3&  3&  4&  4& 2\\
		4&  2&  1&  6& 8\\
		4&  1&  2&  8& 6
	\end{pmatrix},
	\begin{pmatrix}
		1&  1&  1&  1& 1\\
		1&  2&  2&  3& 3\\
		2&  3&  3&  5& 5\\
		2&  4&  4&  7& 7\\
		3&  4&  3&  2& 4\\
		3&  3&  4&  4& 2\\
		4&  2&  1&  8& 6\\
		4&  1&  2&  6& 8
	\end{pmatrix},
	\begin{pmatrix}
		1&  1&  1&  1& 1\\
		1&  3&  2&  3& 3\\
		2&  4&  3&  5& 5\\
		2&  2&  4&  7& 7\\
		3&  4&  4&  2& 8\\
		3&  2&  3&  6& 4\\
		4&  1&  2&  4& 6\\
		4&  3&  1&  8& 2
	\end{pmatrix}.
\end{equation}\\
$\left [ 2:2,2:2 \right ] , \left [ 2:2,2:2 \right ] ,\left [ 2:2,1:1,1:1 \right ] , \left [ 1:1,1:1,1:1,1:1 \right ] , \left [ 1:1,1:1,1:1,1:1 \right ] $:
\begin{equation}\label{43}
	\begin{pmatrix}
		1&  1&  1&  1& 1\\
		1&  2&  3&  3& 3\\
		2&  3&  5&  5& 5\\
		2&  3&  7&  7& 7\\
		3&  4&  6&  2& 4\\
		3&  4&  8&  4& 2\\
		4&  1&  2&  8& 6\\
		4&  2&  4&  6& 8
	\end{pmatrix},
	\begin{pmatrix}
		1&  1&  1&  1& 1\\
		1&  2&  3&  3& 3\\
		2&  3&  5&  5& 5\\
		2&  4&  2&  7& 7\\
		3&  3&  2&  4& 4\\
		3&  4&  4&  8& 2\\
		4&  1&  6&  2& 8\\
		4&  2&  1&  4& 6
	\end{pmatrix},
	\begin{pmatrix}
		1&  1&  1&  1& 1\\
		1&  2&  3&  3& 3\\
		2&  3&  5&  5& 5\\
		2&  4&  2&  7& 7\\
		3&  3&  2&  4& 6\\
		3&  4&  4&  8& 2\\
		4&  1&  6&  2& 8\\
		4&  2&  1&  6& 4
	\end{pmatrix},
\end{equation}
\begin{equation}\label{44}
	\begin{pmatrix}
		1&  1&  1&  1& 1\\
		1&  2&  3&  3& 3\\
		2&  3&  5&  5& 5\\
		2&  4&  3&  7& 7\\
		3&  3&  6&  4& 2\\
		3&  4&  4&  2& 4\\
		4&  1&  2&  8& 6\\
		4&  2&  1&  6& 8
	\end{pmatrix},
	\begin{pmatrix}
		1&  1&  1&  1& 1\\
		1&  2&  2&  3& 3\\
		2&  3&  5&  5& 5\\
		2&  4&  3&  7& 7\\
		3&  3&  6&  2& 4\\
		3&  4&  4&  4& 2\\
		4&  1&  2&  8& 6\\
		4&  2&  1&  6& 8
	\end{pmatrix},
	\begin{pmatrix}
		1&  1&  1&  1& 1\\
		1&  2&  3&  3& 3\\
		2&  3&  5&  5& 5\\
		2&  3&  6&  7& 7\\
		3&  4&  5&  4& 2\\
		3&  4&  6&  2& 4\\
		4&  1&  2&  8& 6\\
		4&  2&  4&  6& 8
	\end{pmatrix},
	\begin{pmatrix}
		1&  1&  1&  1& 1\\
		1&  2&  2&  3& 3\\
		2&  3&  5&  5& 5\\
		2&  4&  3&  7& 7\\
		3&  3&  6&  2& 4\\
		3&  4&  4&  4& 2\\
		4&  1&  2&  8& 6\\
		4&  2&  1&  6& 8
	\end{pmatrix},
	\begin{pmatrix}
		1&  1&  1&  1& 1\\
		1&  2&  2&  3& 3\\
		2&  3&  5&  5& 5\\
		2&  4&  3&  7& 7\\
		3&  3&  6&  4& 2\\
		3&  4&  4&  2& 4\\
		4&  1&  2&  8& 6\\
		4&  2&  1&  6& 8
	\end{pmatrix}.
\end{equation}\\
$\left [ 2:2,2:2 \right ] , \left [ 2:2,2:2 \right ] ,\left [ 2:2,1:1,1:1 \right ] , \left [ 2:2,1:1,1:1 \right ] , \left [ 1:1,1:1,1:1,1:1 \right ] $:
\begin{equation}\label{45}
	\begin{pmatrix}
		1&  1&  1&  1& 1\\
		1&  3&  2&  3& 3\\
		2&  4&  5&  5& 5\\
		2&  2&  3&  6& 7\\
		3&  1&  2&  4& 4\\
		3&  3&  1&  2& 2\\
		4&  4&  4&  6& 6\\
		4&  2&  6&  5& 8
	\end{pmatrix}.
\end{equation}\\
$\left [ 3:1,3:1 \right ] , \left [ 2:2,1:1,1:1 \right ] ,\left [ 2:2,1:1,1:1 \right ] , \left [ 2:2,1:1,1:1 \right ] , \left [ 1:1,1:1,1:1,1:1 \right ] $:
\begin{equation}\label{46}
	\begin{pmatrix}
		1&  1&  1&  1& 1\\
		1&  3&  2&  5& 3\\
		1&  4&  5&  2& 5\\
		2&  5&  3&  3& 7\\
		3&  3&  1&  4& 2\\
		3&  4&  6&  2& 8\\
		3&  6&  2&  1& 4\\
		4&  2&  4&  6& 6
	\end{pmatrix},
	\begin{pmatrix}
		1&  1&  1&  1& 1\\
		1&  2&  5&  3& 3\\
		1&  5&  2&  4& 5\\
		2&  3&  3&  5& 7\\
		3&  4&  2&  3& 4\\
		3&  2&  4&  4& 6\\
		3&  1&  1&  6& 2\\
		4&  6&  6&  2& 8
	\end{pmatrix}.
\end{equation}\\
$\left [ 3:1,3:1 \right ] , \left [ 2:2,1:1,1:1 \right ] ,\left [ 2:2,1:1,1:1 \right ] , \left [ 2:2,1:1,1:1 \right ] , \left [ 2:2,1:1,1:1 \right ] $:
\begin{equation}\label{47}
	\begin{pmatrix}
		1&  1&  1&  1& 1\\
		1&  3&  5&  2& 2\\
		1&  4&  2&  5& 5\\
		2&  5&  3&  3& 3\\
		3&  3&  2&  1& 4\\
		3&  4&  4&  6& 2\\
		3&  6&  1&  2& 1\\
		4&  2&  6&  4& 6
	\end{pmatrix}
\end{equation}\\
$\left [ 2:2,2:2 \right ] , \left [ 2:1,2:1,1:1 \right ] ,\left [ 2:1,2:1,1:1 \right ] , \left [ 2:1,2:1,1:1 \right ] , \left [ 2:1,2:1,1:1 \right ] $:
\begin{equation}\label{48}
	\begin{pmatrix}
		1&  1&  1&  1& 1\\
		1&  2&  3&  3& 3\\
		2&  3&  3&  5& 2\\
		2&  5&  2&  6& 4\\
		3&  6&  5&  4& 2\\
		3&  4&  6&  2& 4\\
		4&  6&  4&  2& 5\\
		4&  4&  2&  4& 6
	\end{pmatrix}.
\end{equation}\\
$\left [ 2:2,2:2 \right ] , \left [ 2:2,1:1,1:1 \right ] ,\left [ 2:1,2:1,1:1 \right ] , \left [ 2:1,2:1,1:1 \right ] , \left [ 1:1,1:1,1:1,1:1 \right ] $:
\begin{equation}\label{49}
	\begin{pmatrix}
		1&  1&  1&  1& 1\\
		1&  3&  5&  2& 3\\
		2&  5&  3&  3& 5\\
		2&  6&  5&  5& 7\\
		3&  6&  6&  6& 2\\
		3&  2&  7&  4& 4\\
		4&  3&  2&  6& 6\\
		4&  4&  4&  2& 8
	\end{pmatrix},
	\begin{pmatrix}
		1&  1&  1&  1& 1\\
		1&  3&  5&  2& 3\\
		2&  5&  3&  3& 5\\
		2&  3&  5&  5& 6\\
		3&  6&  6&  5& 2\\
		3&  2&  4&  6& 4\\
		4&  6&  6&  2& 7\\
		4&  4&  2&  4& 8
	\end{pmatrix},
	\begin{pmatrix}
		1&  1&  1&  1& 1\\
		1&  3&  5&  5& 2\\
		2&  5&  3&  3& 3\\
		2&  3&  5&  5& 4\\
		3&  6&  2&  6& 5\\
		3&  2&  6&  4& 6\\
		4&  6&  2&  6& 7\\
		4&  4&  4&  2& 8
	\end{pmatrix}.
\end{equation}\\
$\left [ 2:2,2:2 \right ] , \left [ 2:2,1:1,1:1 \right ] ,\left [ 2:2,1:1,1:1 \right ] , \left [ 2:1,2:1,1:1 \right ] , \left [ 2:1,2:1,1:1 \right ] $:
\begin{equation}\label{50}
	\begin{pmatrix}
		1&  1&  1&  1& 1\\
		1&  3&  5&  5& 2\\
		2&  5&  3&  3& 2\\
		2&  3&  5&  5& 1\\
		3&  6&  6&  2& 5\\
		3&  2&  4&  6& 6\\
		4&  6&  6&  2& 3\\
		4&  4&  2&  4& 4
	\end{pmatrix}.
\end{equation}\\
$\left [ 2:2,2:2 \right ] , \left [ 2:2,1:1,1:1 \right ] ,\left [ 2:2,1:1,1:1 \right ] , \left [ 1:1,1:1,1:1,1:1 \right ] , \\\left [ 1:1,1:1,1:1,1:1 \right ] $:
\begin{equation}\label{51}
	\begin{pmatrix}
		1&  1&  1&  1& 1\\
		1&  3&  2&  3& 3\\
		2&  5&  5&  5& 5\\
		2&  6&  3&  7& 7\\
		3&  5&  6&  2& 4\\
		3&  6&  4&  4& 2\\
		4&  4&  2&  6& 8\\
		4&  2&  1&  8& 6
	\end{pmatrix},
	\begin{pmatrix}
		1&  1&  1&  1& 1\\
		1&  3&  5&  5& 2\\
		2&  5&  5&  5& 5\\
		2&  6&  3&  7& 7\\
		3&  5&  6&  2& 4\\
		3&  6&  4&  4& 2\\
		4&  4&  2&  8& 6\\
		4&  2&  1&  6& 8
	\end{pmatrix}.
\end{equation}\\
$\left [ 2:2,2:2 \right ] , \left [ 2:2,1:1,1:1 \right ] ,\left [ 2:2,1:1,1:1 \right ] , \left [ 2:2,1:1,1:1 \right ] , \left [ 1:1,1:1,1:1,1:1 \right ] $:\\
\begin{equation}\label{52}
	\begin{pmatrix}
		1&  1&  1&  1& 1\\
		1&  3&  3&  2& 3\\
		2&  5&  5&  5& 5\\
		2&  7&  2&  3& 7\\
		3&  6&  2&  6& 4\\
		3&  8&  4&  4& 2\\
		4&  4&  1&  2& 6\\
		4&  2&  6&  1& 8
	\end{pmatrix}.\\
\end{equation}\\
$\left [ 3:2,2:1 \right ] , \left [ 2:2,1:1,1:1 \right ] ,\left [ 2:1,2:1,1:1 \right ] , \left [ 2:1,2:1,1:1 \right ] , \left [ 1:1,1:1,1:1,1:1 \right ] $:\\
\begin{equation}\label{53}
	\begin{pmatrix}
		1&  1&  1&  1& 1\\
		1&  5&  5&  3& 2\\
		1&  6&  2&  5& 3\\
		2&  3&  5&  2& 5\\
		2&  4&  3&  3& 7\\
		3&  5&  2&  4& 6\\
		3&  6&  4&  2& 4\\
		4&  2&  6&  6& 8
	\end{pmatrix}.\\
\end{equation}
$\left [ 2:2,1:1,1:1 \right ] , \left [ 2:2,1:1,1:1 \right ] ,\left [ 2:2,1:1,1:1 \right ] , \left [ 2:1,2:1,1:1 \right ] , \left [ 2:1,2:1,1:1 \right ] $:
\begin{equation}\label{54}
	\begin{pmatrix}
		1&  1&  1&  1& 1\\
		3&  5&  3&  2& 1\\
		5&  3&  5&  5& 2\\
		6&  1&  3&  6& 2\\
		7&  6&  5&  2& 5\\
		2&  4&  4&  3& 6\\
		8&  3&  2&  1& 3\\
		4&  2&  6&  4& 4
	\end{pmatrix},
	\begin{pmatrix}
		1&  1&  1&  1& 1\\
		3&  2&  1&  1& 1\\
		5&  5&  2&  5& 3\\
		6&  3&  2&  5& 3\\
		6&  4&  5&  2& 5\\
		4&  2&  6&  3& 4\\
		5&  6&  3&  4& 2\\
		2&  1&  4&  6& 6
	\end{pmatrix}.
\end{equation}

Next, we introduce the method we use to find these $5$-qubit UOMs of size eight in FIG \ref{fig:flowchart}. 
\begin{figure}[H]
	\centering
	\begin{minipage}[t]{0.4\textwidth}
		\scalebox{0.9}{%
			\begin{tikzpicture}[node distance=2cm, auto, on grid]
				\tikzset{
					startstop/.style={
						rectangle,
						rounded corners,
						minimum width=6cm,
						minimum height=1cm,
						align=center,
						draw=black,
						fill=white!30
					},
					io/.style={
						rectangle,
						minimum width=3cm,
						minimum height=1cm,
						align=center,
						draw=black,
						fill=white!30
					},
					process/.style={
						rectangle,
						minimum width=3cm,
						minimum height=1cm,
						align=center,
						draw=black,
						fill=white!30
					},
					decision/.style={
						diamond,
						minimum width=3cm,
						minimum height=1cm,
						align=center,
						draw=black,
						fill=white!30
					},
					arrow/.style={
						thick,
						->,
						>=stealth
					}
				}
				
				\node (start) [startstop,text width=8cm] {(i) Start};
				\node (input) [io, below=1.3cm of start,text width=8cm] {(ii) Inputting the size of the UOM that you want};
				\node (process1) [process, below=1.4cm of input,text width=8cm] {(ii) Using Theorem \ref{eq:5} to Calculate all the possible column feature lists that can satisfy the orthogonality};
				\node (process2) [process, below=1.5cm of process1,text width=8cm] {(iii) Preliminary examination of all the possible column feature lists};
				\node (process2a) [process, below =2cm of process2,text width=8cm]  {(iv) For every possible column feature list we obtain, checking any two columns of them. Obtaining linear integer diophantine equations. And applying the methods in \cite{frumkin1977polynomial,aardal2000solving,ruzsa1993solving} to solve them. Obtaining all the possible columns that may form a matrix};
				\node (process2c) [process, below= 2.2cm of process2a,text width=8cm] {(v) Using the Theorem \ref{the:UOM} to check all the columns we obtain and construct UOMs. Using the algorithm in section \ref{sec:85 upb} to check all the UOMs we obtain};
				\node (stop) [startstop, below= 1.6cm of process2c,text width=8cm] {Stop};
				
				\draw [arrow] (start) -- (input);
				\draw [arrow] (input) -- (process1);
				\draw [arrow] (process1) -- (process2);
				\draw [arrow] (process2) -- (process2a);
				\draw [arrow] (process2a) -- (process2c);
				\draw [arrow] (process2c) -- (stop);
				
			\end{tikzpicture}
		}
	\end{minipage}
	\captionsetup{justification=centering}
	\caption{Program Flowchart and Step Explanation.}
	\label{fig:flowchart}
\end{figure}
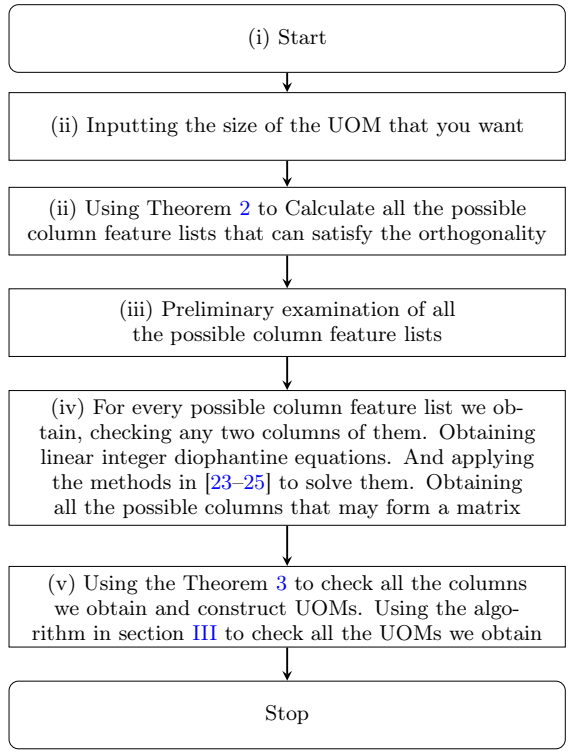\par
By observing the $5$-qubit UOMs of size eight, we obtain a theorem. One can use this theorem to construct a new UOM non-equivalent to a given one. 
\begin{theorem}\label{the:2123}
	We can construct a new UOM non-equivalent to a given UOM, if the UOM satisfies
	
	\par
	(i) A submatrix of the form $\begin{pmatrix}
		a& c\\
		c& a
	\end{pmatrix}$ appears in two columns of the UOM. The multiplicity of $a$ and $c$ are both $1$ in these two columns such that $a \ne c'$ and $a \ne c$.
	\par
	(ii) The $a'$s appear in the same rows in these two columns. The  $c'$s appear in the same rows in these two columns. Then, we can construct a new UOM by exchanging the two columns of $\begin{pmatrix}
		c& a\\
		a& c
	\end{pmatrix}$.

\end{theorem}
\begin{proof}
	Because of Theorem \ref{the:2123} (i), the rows of the UOM after exchanging the two columns of the submatrix  $\begin{pmatrix}
		a& c\\
		c& a
	\end{pmatrix}$ will still remain pairwise orthogonal. We will proof why the UOM we obtain is still unextendible.
	If the UOM we obtain is not a UOM, then because of Theorem \ref{the:UOM}, there exist a set $\left \{ w_i \right \}$ that the inequality does not hold. We have two cases $(a)$ and $(b)$ as follows.\par
	(a) Any $w_i$ in the set is not generated by $a$ and $c$ in these two columns. It is contradictive because the set can be generated from the UOM before exchanging.\par
	(b) If a $w_i$ is generated by $a$ and $c$ in these two columns, then because of Theorem \ref{the:2123} (ii), the multiplicity of $a$ and $c$ are both $1$ in these two columns, so the $w_i$ is 1, the UOM before exchanging can generate the same $w_i$.	 
\end{proof}

\section{conclusion}
\label{sec:con}
 We have designed a method to check if two UPBs are non-equivalent to each other. We have designed a method to find all non-equivalent UPBs of a distinct size and given all $8\times5$ UPBs as an example. This method would be helpful to other sizes too.
Our results have made some methods to help finding UPBs of a given size. It will be inefficient if the size of UPBs is large. Actually, we are not sure if there exists a $9$-qubit UPB of size $13$.


\section*{Acknowledgements}

Authors were supported by the NNSF of China (Grant No. 11871089).

\bibliographystyle{unsrt}
\bibliography{cankao}

\end{document}